\documentclass[12pt]{article}

\newtheorem{proposition}{Proposition}
\newtheorem{theorem}{Theorem}
\newtheorem{corollary}{Corollary}
\newtheorem{lemma}{Lemma}
\newenvironment{proof}{\noindent\textbf{Proof.}}{\hfill\rule{2mm}{2mm}\medskip}

\def\dd{\mathinner{\ldotp\ldotp}}   
\def\per{\mathit{}{per}}               
\def\exp{\mathit{exp}}               
\def\beg{\mathit{beg}}
\def\endd{\mathit{end}}
\def\minleaf{\mathit{minleaf}}
\def\firstocc{\mathit{firstocc}}
\def\alignocc{\mathit{alignocc}}
\def\factends{\mathit{factends}}
\def\rep{\pi{}}

\title{Optimal searching of gapped repeats in a word}

\author{
Maxime Crochemore%
\thanks{King's College London,
London WC2R 2LS, UK
and Universit{\'e} Paris-Est, France, {\tt Maxime.Crochemore@kcl.ac.uk}}
\and
Roman Kolpakov\thanks{Lomonosov Moscow State University, Leninskie Gory, Moscow, 119992 Russia,
{\tt foroman@mail.ru}}
\and
Gregory Kucherov%
\thanks{LIGM, Universit{\'e} Paris-Est Marne-la-Vall{\'e}e,
77454 Marne-la-Vall{\'e}e CEDEX 2, France,
{\tt Gregory.Kucherov@univ-mlv.fr}}
}
\date{\empty}

\begin{document}
\maketitle



\begin{abstract}
Following (Kolpakov et al., 2013; Gawrychowski and Manea, 2015),
we continue the study of {\em $\alpha$-gapped repeats}
in strings, defined as factors $uvu$ with
$|uv|\leq \alpha |u|$. Our main result is the $O(\alpha n)$ bound on
the number of {\em maximal} $\alpha$-gapped repeats in a string of
length $n$, previously proved to be $O(\alpha^2 n)$ in (Kolpakov et al., 2013).
For a closely
related notion of maximal $\delta$-subrepetition (maximal factors of exponent between $1+\delta$ and $2$), our result implies the $O(n/\delta)$ bound on their number, which improves the bound of (Kolpakov et al., 2010)
by a $\log n$ factor.

We also prove an algorithmic time bound $O(\alpha n+S)$ ($S$ size of
the output) for computing all maximal $\alpha$-gapped repeats. Our
solution, inspired by (Gawrychowski and Manea, 2015),
is different from the recently
published proof by (Tanimura et al., 2015)
of the same bound. Together with our bound on $S$, this implies an
$O(\alpha n)$-time algorithm for computing all maximal $\alpha$-gapped
repeats.

\end{abstract}

\section{Introduction}

\paragraph{Notation and basic definitions.}
Let $w=w[1]w[2]\ldots w[n]=w[1 \dd n]$ be an arbitrary word.
The length $n$
of~$w$ is denoted by $|w|$. For any $1\le i\le j\le n$, word $w[i]\ldots
w[j]$ is called a {\it factor} of~$w$ and is denoted by $w[i \dd j]$. Note
that notation $w[i \dd j]$ denotes two entities: a word and its
occurrence starting at position $i$ in $w$.
To underline the second meaning, we will sometimes use the term {\it segment}.
Speaking about the equality between factors can also be ambiguous, as it may mean that the factors are identical words or identical segments.
If two factors $u,
v$ are identical words, we call them {\it equal} and denote this by
$u=v$. To express that $u$ and $v$ are the same segment, we use the notation
$u\equiv v$. For any $i=1\ldots n$, factor $w[1 \dd i]$ (resp. $w[i \dd
n]$) is a {\it prefix} (resp. {\it suffix}) of~$w$. By {\em positions}
on~$w$ we mean indices $1, 2,\ldots ,n$ of letters in $w$. For
any factor~$v\equiv w[i \dd j]$ of~$w$, positions $i$ and $j$ are called
respectively {\it start position} and {\it end position} of~$v$ and denoted by
$\mathit{beg}(v)$ and $\mathit{end}(v)$ respectively. Let $u, v$ be two factors
of~$w$. Factor $u$ {\it is contained} in~$v$ iff $\mathit{beg}(v)\le \mathit{beg}(u)$ and $\mathit{end}(u)\le \mathit{end}(v)$. Letter $w[i]$ {\it is contained} in~$v$ iff $\mathit{beg}(v)\le i\le \mathit{end}(v)$.

A positive integer $p$ is called  a {\it period} of~$w$ if $w[i]=w[i+p]$ for
each $i=1,\ldots ,n-p$. We denote by $\per(w)$ the {\em smallest period} of~$w$ and
define the {\it exponent} of~$w$ as $\exp(w)=|w|/\per(w)$.
A word is called
{\it periodic} if its exponent is at least~2. Occurrences of periodic words
are called {\it repetitions}.

\paragraph{Repetitions, squares, runs.}
Patterns in strings formed by repeated factors are of primary importance
in word combinatorics~\cite{Lothaire83} as well as
in various applications such as string matching
algorithms~\cite{GaliSeiferas83,CrochRytter95}, molecular
biology~\cite{Gusfield97}, or text compression~\cite{Storer88}. The simplest
and best known example of such patterns is a factor of the form $uu$, where $u$
is a nonempty word. Such repetitions are called {\it squares}.
Squares have been extensively studied.
While the number of all square occurrences can be quadratic (consider
word $\texttt{a}^n$), it is known that the number of {\em primitively-rooted}
squares is $O(n\log n)$ \cite{CrochRytter95}, where a square $uu$ is primitively-rooted if
the exponent of $u$ is not an integer greater than~$1$.
An optimal $O(n\log n)$-time algorithm for
finding all primitively-rooted squares was proposed in~\cite{Crochemor81}.

Repetitions can be seen as a natural generalization of squares.
A repetition in a given word is called {\it maximal} if it cannot
be extended by at least one
letter to the left nor to the right without changing (increasing) its minimal period. More precisely, a repetition
$r\equiv w[i \dd j]$ in~$w$ is called {\it maximal} if it satisfies the
following conditions:
\begin{enumerate}
\item $w[i-1]\neq w[i-1+\per(r)]$ if $i>1$,
\item $w[j+1-\per(r)]\neq w[j+1]$ if $j<n$.
\end{enumerate}
For example, word $\texttt{cababaaa}$ has two maximal repetitions: $\texttt{ababa}$
 and $\texttt{aaa}$. Maximal repetitions are usually called {\it runs} in the literature. Since any repetition is contained in some run,
the set of all runs can be
considered as a compact encoding of all repetitions in the word. This set has many useful applications, see, e.g., \cite{Crochetal1}. For any
word~$w$, we denote by ${\cal R}(w)$ the number of maximal repetitions in~$w$ and by ${\cal E}(w)$ the sum of exponents of all maximal repetitions in~$w$. The following statements are proved in~\cite{KK00}.
\begin{theorem}
$\max_{|w|=n}{\cal E}(w)=O(n)$. \label{sumexp}
\end{theorem}
\begin{corollary}
$\max_{|w|=n}{\cal R}(w)=O(n)$.
\label{onmaxrun}
\end{corollary}
A series of
papers (e.g.,~\cite{CrochIlieTinta, Crochetal11}) focused on more precise upper bounds on ${\cal E}(w)$
and ${\cal R}(w)$ trying to obtain the best possible constant factor behind the $O$-notation. A breakthrough in this direction was recently made in~\cite{RunsTheor} where the so-called
``runs conjecture'' ${\cal R}(w[1..n])<n$ was proved. To the best of our knowledge, the currently best upper bound
${\cal R}(w[1..n])\le\frac{22}{23}n$ on ${\cal R}(w)$ is shown
in~\cite{BeyRunsTheor}.

On the algorithmic side, an $O(n)$-time algorithm for finding all runs in a
word of length~$n$ was proposed in~\cite{KK00} for the case of constant-size alphabet.
Another $O(n)$-time algorithm, based on a different approach, has been
proposed in~\cite{RunsTheor}. The $O(n)$ time bound holds for the
(polynomially-bounded) integer alphabet as well, see, e.g.,
\cite{RunsTheor}. However, for the case of unbounded-size alphabet
where characters can only be tested for equality, the lower bound
$\Omega(n\log n)$ on computing all runs has been known for a long time
\cite{MainLorentz85}. It is an interesting open question (raised over
20 years ago in \cite{BreslauerPhD92}) whether the $O(n)$ bound holds 
for an unbounded linearly-ordered alphabet. Some results related to this 
question have recently been obtained in \cite{KosolobovSTACS15}.

\paragraph{Gapped repeats and subrepetitions.}
Another natural generalization of squares are factors of the form $uvu$ where $u$
and $v$ are nonempty words. We call such factors {\it gapped repeats}. For a
gapped repeat $uvu$, the left (resp. right) occurrence of~$u$ is called the {\it left}
(resp. {\em right}) {\em copy}, and $v$ is called the {\it gap}. The {\it period} of this
gapped repeat is $|u|+|v|$. For a gapped repeat~$\rep$,
we denote the length of copies of~$\rep$ by $c(\rep)$ and the period
of~$\rep$ by $p(\rep)$.
Note that a gapped repeat $\rep=uvu$ may have different periods, and $\per(\rep)\leq p(\rep)$.
For example, in string $\texttt{cabacaabaa}$, segment $\texttt{abacaaba}$
 corresponds to two gapped repeats having copies $\texttt{a}$ and $\texttt{aba}$
 and periods $7$ and $5$ respectively. Gapped repeats forming the same segment but having different periods are considered distinct.
This means that to specify a gapped repeat it is generally not sufficient to specify its segment.
If $u',u''$ are equal non-overlapping factors and $u'$ occurs to the left of $u''$, then by $(u',u'')$ we denote the gapped repeat with left copy~$u'$ and right copy~$u''$.
For a given gapped
repeat $(u', u'')$, equal factors $u'[i  \dd
j]$ and $u''[i  \dd j]$, for $1\le i\le j\le |u'|$, of the copies $u'$, $u''$ are called {\it corresponding
factors} of repeat $(u', u'')$.

For any real $\alpha>1$, a gapped repeat~$\rep$ is called {\it
$\alpha$-gapped} if $p(\rep)\le\alpha c(\rep)$. Maximality of gapped
repeats is defined similarly to
repetitions. A
gapped repeat $(w[i' \dd j'], w[i'' \dd j''])$ in~$w$ is called {\it maximal}
if it satisfies the following conditions:
\begin{enumerate}
\item $w[i'-1]\neq w[i''-1]$ if $i'>1$,
\item $w[j'+1]\neq w[j''+1]$ if $j''<n$.
\end{enumerate}
In other words, a gapped repeat $\rep$ is maximal if its copies cannot be
extended to the left nor to the right by at least one letter
without breaking its period $p(\rep)$. As observed in~\cite{KPPHr13},
any $\alpha$-gapped repeat is contained either in a (unique) maximal
$\alpha$-gapped repeat with the same period, or in a (unique) maximal
repetition with a period which is a divisor of the repeat's period. For example,
 in the above string $\texttt{cabacaabaa}$, gapped repeat $(\texttt{ab})\texttt{aca}(\texttt{ab})$
 is contained in maximal repeat $(\texttt{aba})\texttt{ca}(\texttt{aba})$
 with the same period $5$. In string $\texttt{cabaaabaaa}$, gapped repeat
 $(\texttt{ab})\texttt{aa}(\texttt{ab})$ with period $4$ is contained in maximal 
 repetition $\texttt{abaaabaaa}$ with period $4$. Since all maximal repetitions 
 can be computed efficiently in $O(n)$ time (see above),
the problem of computing all $\alpha$-gapped repeats in a word can be
reduced to the problem of finding all maximal $\alpha$-gapped repeats.

Several variants of the problem of computing gapped repeats have been studied earlier. In~\cite{Brodal00}, it was shown that all
maximal gapped repeats with a gap length belonging to a specified interval can
be found in time $O(n\log n+S)$, where $n$ is the word length and $S$ is
output size. In \cite{KK00a}, an algorithm was proposed for finding all gapped repeats with a fixed gap
length $d$ running in time
$O(n\log d+S)$.
In~\cite{KPPHr13}, it was proved that the number of maximal
$\alpha$-gapped repeats  in a word of length~$n$ is bounded by $O(\alpha^2
n)$ and all maximal $\alpha$-gapped repeats can be found in
$O(\alpha^2 n)$ time for the case of integer alphabet. A new approach to
computing gapped repeats was recently proposed in~\cite{GabrMan, DumiMan}. In particular,
in~\cite{GabrMan} it is shown that the longest $\alpha$-gapped repeat in a
word of length~$n$ over an integer alphabet can be found in $O(\alpha n)$
time. Finally, in a recent paper~\cite{Tanimuraetal}, an algorithm is proposed for finding all
maximal $\alpha$-gapped repeats in $O(\alpha n+S)$ time where $S$ is the output size, for a constant-size alphabet.
The algorithm uses an approach previously introduced in \cite{BadkobehCrochToop12}.

Recall that repetitions are segments with exponent at least $2$. Another way to approach gapped repeats is to consider segments with exponent smaller than $2$, but strictly greater than $1$. Clearly, such a segment corresponds to a gapped repeat $\rep=uvu$ with $\per(\rep)=p(\rep)=|u|+|v|$. We will call such factors (segments) {\em subrepetitions}.
More
precisely, for any~$\delta$, $0<\delta<1$, by a $\delta$-subrepetition we
mean a factor~$v$ that satisfies $1+\delta\le \exp(v)<2$. Again, the
notion of maximality straightforwardly applies to
subrepetitions as well: maximal subrepetitions are defined exactly in the same way as
maximal repetitions.
The relationship between maximal subrepetitions and maximal gapped repeats was clarified in \cite{KPPHr13}. Directly from the definitions, a maximal subrepetition $\rep$ in a string $w$ corresponds to a maximal gapped repeat with $p(\rep)=\per(\rep)$.
Futhermore, a maximal $\delta$-subrepetition corresponds to a maximal
$\frac{1}{\delta}$-gapped repeat. However, there may be more maximal
$\frac{1}{\delta}$-gapped repeats than maximal
$\delta$-subrepetitions, as 
not every maximal $\frac{1}{\delta}$-gapped repeat corresponds to a
maximal $\delta$-subrepetition. 

Some combinatorial results on the number of maximal
subrepetitions in a string were obtained in~\cite{KKOch}. In particular, it was
proved that the number of maximal $\delta$-subrepetitions in a word of
length~$n$ is bounded by $O(\frac{n}{\delta}\log n)$.
In~\cite{KPPHr13},
an $O(n/\delta^2)$ bound on
the number of maximal $\delta$-subrepetitions in a word of length~$n$ was
obtained. Moreover, in~\cite{KPPHr13}, two algorithms were proposed for finding all maximal
$\delta$-subrepetitions in the word running respectively in
$O(\frac{n\log\log n}{\delta^2})$ time and in $O(n\log
n+\frac{n}{\delta^2}\log\frac{1}{\delta})$ expected time, over the integer alphabet. In \cite{BadkobehCrochToop12}, it is shown that all subrepetitions with the largest exponent (over all subrepetitions) can be found in an overlap-free string in time $O(n)$, for a constant-size alphabet.

\paragraph{Our results.}
In the present work we improve the results of~\cite{KPPHr13} on
maximal gapped repeats: we prove an asymptotically tight bound of $O(\alpha
n)$ on the number of maximal $\alpha$-gapped repeats in a word of
length~$n$ (Section~\ref{counting}). From our bound, we also derive a $O(n/\delta)$ bound on
the number of maximal $\delta$-subrepetitions occurring in the word,
which improves the bound of \cite{KKOch} by a $\log n$ factor.
Then, based on the algorithm of~\cite{GabrMan}, we obtain
an asymptotically optimal $O(\alpha n)$ time bound for computing all maximal
$\alpha$-gapped repeats in a string (Section~\ref{algorithm}). Note that this bound follows from the recently published paper \cite{Tanimuraetal} that presents an $O(\alpha n+S)$ algorithm for computing all maximal $\alpha$-gapped repeats. Here we present an alternative algorithm with the same bound that we obtained independently.

\section{Preliminaries}

In this Section we state a few propositions that will be used later in
the paper. The following fact is well-known (see,
e.g.,~\cite[Proposition~2]{Kolpakov12}).

\begin{proposition}
Any period~$p$ of a word~$v$ such that $|v|\ge 2p$ is divisible by $\per(v)$,
the smallest period of $v$. \label{divminper}
\end{proposition}

Let $\Delta$ be some natural number. A period~$p$ of some word~$v$
is called {\it $\Delta$-period} if $p$ is divisible by~$\Delta$.
The minimal $\Delta$-period of~$v$, if exists, is denoted by $p_{\Delta}(v)$.
The word~$v$ is called $\Delta$-periodic if $|v|\ge 2p_{\Delta}(v)$.
It is obvious that any $\Delta$-periodic word is also periodic.
Proposition~\ref{divminper} can be generalized in the following way.

\begin{proposition}
Any $\Delta$-period~$p$ of a word~$v$ such that $|v|\ge 2p$
is divisible by $p_{\Delta}(v)$.
\label{minDeltaper}
\end{proposition}

\begin{proof}
By Proposition~\ref{divminper}, period~$p$ is divisible by $\per(v)$,
so $p$ is divisible by $LCM(\per(v),\Delta)$. On the other hand,
$LCM(\per(v),\Delta)$ is a $\Delta$-period of~$v$. Thus,
$p_{\Delta}(v)=LCM(\per(v),\Delta)$, and $p$ is divisble by $p_{\Delta}(v)$.
\end{proof}

Consider an arbitrary word $w=w[1 \dd n]$ of
length~$n$. Recall that any repetition~$y$ in~$w$ is extended to a unique
maximal repetition $r$ with the same minimal period. We call $r$ the {\it extension} of~$y$.

Let $r$ be a repetition in the word~$w$. We call any factor of~$w$ of
length $\per(r)$ which is contained in~$r$ a {\it cyclic root} of~$r$. For
cyclic roots we have the following property proved, e.g.,
in~\cite[Proposition~2]{KPPHr13}.

\begin{proposition}
Two cyclic root $u'$, $u''$ of a repetition~$r$ are equal if and only if
$\beg(u')\equiv \beg(u'')\pmod{\per(r)}$.
\label{oncycroot}
\end{proposition}

\section{Number of maximal repeats and subrepetitions}
\label{counting}

In this section, we obtain an improved upper bound on the number of maximal gapped repeats and
subrepetitions in a string $w$.
Following the general approach of \cite{KPPHr13}, we split all maximal gapped repeats into three categories according
to periodicity properties of repeat's copy: periodic, semiperiodic and
ordinary repeats. Bounds for periodic and semiperiodic repeats
are directly borrowed from \cite{KPPHr13}, while for ordinary repeats, we
obtain a better bound.

\paragraph{\it Periodic repeats.}
We say that a maximal gapped
repeat is {\it periodic} if its copies are periodic strings (i.e. of exponent at least $2$). The set
of all periodic maximal $\alpha$-gapped repeats in $w$ is denoted by
${\cal PP}_{\alpha}$. The following bound on
the size of ${\cal PP}_{\alpha}$ was been obtained in \cite[Corollary~6]{KPPHr13}.

\begin{lemma}
$|{\cal PP}_k|=O(kn)$ for any natural $k>1$.
\label{onPP}
\end{lemma}

\paragraph{\it Semiperiodic repeats.}
A maximal gapped repeat is called {\it prefix} ({\it suffix}) {\it semi\-periodic}
if the copies of this repeat are not periodic, but have a prefix (suffix)
which is periodic and its length is at least half of the copy length.
A maximal gapped repeat is {\it semiperiodic} if it is either prefix
or suffix semiperiodic. The set of all semiperiodic $\alpha$-gapped maximal
repeats is denoted by~${\cal SP}_{\alpha}$.
In~\cite[Corollary~8]{KPPHr13}, the following bound was obtained on the number
of semiperiodic maximal $\alpha$-gapped repeats.

\begin{lemma}[\cite{KPPHr13}]
$|{\cal SP}_k|=O(kn)$ for any natural $k>1$.
\label{onSP}
\end{lemma}

\paragraph{\it Ordinary repeats.}
Maximal gapped repeats which are neither periodic nor semiperiodic are called
{\it ordinary}. The set of all ordinary maximal $\alpha$-gapped repeats in
the word~$w$ is denoted by~${\cal OP}_{\alpha}$. In the rest of this
section, we prove that the cardinality of ${\cal OP}_{\alpha}$ is
$O(\alpha n)$.
%
For simplicity, assume that $\alpha$ is an
integer number $k$.

To estimate the number of ordinary maximal $k$-gapped repeats, we
use the following idea from~\cite{Kolpakov12}.
We represent
a maximal repeat $\rep\equiv (u', u'')$ from ${\cal OP}_k$ by a triple
$(i, j, c)$ where $i=\beg(u')$, $j=\beg(u')$ and
$c=c(\rep)=|u'|=|u''|$. Such triples will be called {\em points}.
Obviously, $\rep$ is uniquely defined by
values $i$, $j$ and~$c$, therefore two different repeats from ${\cal OP}_k$ can
not be represented by the same point.

For any two points $(i', j', c')$, $(i'', j'', c'')$ we say that point
$(i', j', c')$ {\it covers} point $(i'', j'', c'')$ if $i'\le i''\le
i'+c'/6$, $j'\le j''\le j'+c'/6$, $c'\ge c'' \ge \frac{2c'}{3}$.
A point is {\it covered} by a repeat $\rep$ if
this it is covered by the point representing~$\rep$. By $V[\rep ]$ we
denote the set of all points covered by a repeat~$\rep$. We show that any
point can not be covered by two different repeats from ${\cal OP}_k$.

\begin{lemma}
Two different repeats from ${\cal OP}_k$ cannot cover the same point.
\label{keylemma}
\end{lemma}

\begin{proof}
Let $\rep_1\equiv (u'_1, u''_1)$, $\rep_2\equiv (u'_2, u''_2)$ be two
different repeats from ${\cal OP}_k$ covering the same point $(i, j, c)$.
Denote $c_1=c(\rep_1)$, $c_2=c(\rep_2)$, $p_1=\per(\rep_1)$,
$p_2=\per(\rep_2)$. Without loss of generality we assume $c_1\ge c_2$. From
$c_1\ge c\ge \frac{2c_1}{3}$, $c_2\ge c\ge \frac{2c_2}{3}$ we have $c_1\ge
c_2\ge \frac{2c_1}{3}$, i.e. $c_2\le c_1\le\frac{3c_2}{2}$. Note that $w[i]$
is contained in both left copies $u'_1, u'_2$, i.e. these copies
overlap. If $p_1=p_2$, then repeats $\rep_1$ and
$\rep_2$ must coincide due to the maximality of these
repeats. Thus, $p_1\neq p_2$. Denote $\Delta =|p_1-p_2|>0$. From $\beg(u'_1)\le i\le \beg(u'_1)+c_1/6$ and $\beg(u''_1)\le j\le \beg(u''_1)+c_1/6$ we have
$$
(j-i)-c_1/6\le p_1\le (j-i)+c_1/6.
$$
Analogously, we have
$$
(j-i)-c_2/6\le p_2\le (j-i)+c_2/6.
$$
Thus $\Delta\le (c_1+c_2)/6$ which, together with inequality $c_1\le\frac{3c_2}{2}$,
implies $\Delta\le\frac{5c_2}{12}$.

First consider the case when one of the copies $u'_1, u'_2$ is contained in
the other, i.e. $u'_2$ is contained in $u'_1$. In this case, $u''_1$
contains some factor $\widehat u''_2$ corresponding to the factor $u'_2$ in
$u'_1$. Since $\beg(u''_2)-\beg(u'_2)=p_2$, $\beg(\widehat
u''_2)-\beg(u'_2)=p_1$ and $u''_2=\widehat u''_2=u'_2$, we have
$$
|\beg(u''_2)-\beg(\widehat u''_2)|=\Delta,
$$
so $\Delta$ is a period of $u''_2$ such that
$\Delta\le\frac{5}{12}c_2=\frac{5}{12}|u''_2|$. Thus, $u''_2$ is periodic
which contradicts that $\rep_2$ is not periodic.

Now consider the case when $u'_1, u'_2$ are not contained in one another.
Denote by $z'$ the overlap of  $u'_1$ and $u'_2$. Let $z'$ be a suffix of
$u'_k$ and a prefix of $u'_l$ where $k,l=1, 2$, $k\neq l$. Then $u''_k$
contains a suffix $z''$ corresponding to the suffix $z'$ in $u'_k$, and
$u''_l$ contains a prefix $\widehat z''$ corresponding to the prefix $z'$ in
$u'_l$. Since $\beg(z'')-\beg(z')=p_k$ and $\beg(\widehat
z'')-\beg(z')=p_l$ and $z''=\widehat z''=z'$, we have
$$
|\beg(z'')-\beg(\widehat z'')|=|p_k-p_l|=\Delta,
$$
therefore $\Delta$ is a period of $z'$. Note that in this case
$$
\beg(u'_k)<\beg(u'_l)\le i \le \beg(u'_k)+c_k/6,
$$
therefore $0<\beg(u'_l)-\beg(u'_k)\le c_k/6$. Thus
$$
|z'|=c_k-(\beg(u'_l)-\beg(u'_k))\ge\frac{5}{6}c_k\ge\frac{5}{6}c_2.
$$
From $\Delta\le\frac{5}{12}c_2$ and $c_2\le\frac{6}{5}|z'|$ we obtain
$\Delta\le |z'|/2$. Thus, $z'$ is a periodic suffix of $u'_k$ such that
$|z'|\ge \frac{5}{6}|u'_k|$, i.e. $\rep_k$ is either suffix semiperiodic or
periodic which contradicts $\rep_k\in {\cal OP}_k$.
\end{proof}

Denote by ${\cal Q}_k$ the set of all points $(i, j, c)$ such that $1\le i,
j, c\le n$ and $i<j\le i+(\frac{3}{2}k+\frac{1}{4})c$.

\begin{lemma}
Any point covered by a repeat from ${\cal OP}_k$ belongs to ${\cal Q}_k$.
\label{keylemma1}
\end{lemma}

\begin{proof}
Let a point $(i, j, c)$ be covered by some repeat $\rep\equiv (u', u'')$
from ${\cal OP}_k$. Denote $c'=c(\rep )$. Note that $w[i]$ and $w[j]$ are
contained respectively in $u'$ and $u''$ and $n>c'\ge c\ge \frac{2c'}{3}>0$,
so inequalities $1\le i, j, c\le n$ and $i<j$ are obvious. Note also that
$$
j\le \beg(u'')+c'/6=\beg(u')+\per(\rep )+c'/6\le i+kc'+c'/6,
$$
therefore, taking into account $c'\le \frac{3c}{2}$, we have $j\le
i+(\frac{3}{2}k+\frac{1}{4})c$.
\end{proof}

From Lemmas~\ref{keylemma} and~\ref{keylemma1},  we obtain

\begin{lemma}
$|{\cal OP}_k|=O(nk)$.
\label{OPklemma}
\end{lemma}

\begin{proof}
Assign to each point $(i, j, c)$ the weight $\rho (i,
j, c)=1/c^3$. For any finite set~$A$ of points, we define
$$
\rho (A)=\sum_{(i, j, c)\in A} \rho (i, j, c)=\sum_{(i, j, c)\in A}\frac{1}{c^3}.
$$
Let $\rep$ be an arbitrary repeat from ${\cal OP}_k$ represented by a point
$(i', j', c')$. Then
\begin{eqnarray*}
\rho (V[\rep ])&=&\sum_{i'\le i\le i'+c'/6}\;\sum_{j'\le j\le j'+c'/6}\;\sum_{2c'/3\le c\le c'}\frac{1}{c^3}\\
&>&\frac{c'^2}{36}\sum_{2c'/3\le c\le c'}\frac{1}{c^3}.
\end{eqnarray*}
Using a standard estimation of sums by integrals, one can deduce that
$\sum_{2c'/3\le c\le c'}\frac{1}{c^3}\ge \frac{5}{32}\frac{1}{c'^2}$ for any~$c'$.
Thus, for any $\rep$ from ${\cal OP}_k$
$$
\rho (V[\rep ])>\frac{1}{36}\frac{5}{36}=\Omega (1).
$$
Therefore,
\begin{equation}
\sum_{\rep\in {\cal OP}_k}\rho (V[\rep ])=\Omega (|{\cal OP}_k|).
\label{lowbndOP}
\end{equation}
Note also that
\begin{eqnarray*}
\rho ({\cal Q}_k)&\le &\sum_{i=1}^n\;\sum_{i<j\le i+(\frac{3}{2}k+\frac{1}{4})c}\;
\sum_{c=1}^n\frac{1}{c^3}\\
&<&n(\frac{3}{2}k+\frac{1}{4})c\sum_{c=1}^n\frac{1}{c^3}<2nk\sum_{c=1}^n\frac{1}{c^2}
<2nk\sum_{c=1}^{\infty}\frac{1}{c^2}=\frac{nk\pi^2}{3}.
\end{eqnarray*}
Thus,
\begin{equation}
\rho ({\cal Q}_k)=O(nk).
\label{upbndQ}
\end{equation}
By Lemma~\ref{keylemma1}, any point covered by repeats from ${\cal OP}_k$
belongs to ${\cal Q}_k$. On the other hand, by Lemma~\ref{keylemma}, each
point of ${\cal Q}_k$ can not be covered by two repeats from ${\cal OP}_k$.
Therefore,
 $$
\sum_{\rep\in {\cal OP}_k}\rho (V[\rep ])\le \rho ({\cal Q}_k).
$$
Thus, using \ref{lowbndOP} and~\ref{upbndQ}, we conclude that $|{\cal
OP}_k|=O(nk)$.
\end{proof}

Putting together Lemma~\ref{onPP}, Lemma~\ref{onSP}, and Lemma~\ref{OPklemma}, we
obtain that for any integer $k\ge 2$, the number of maximal $k$-gapped repeats
in~$w$ is $O(nk)$. The bound straightforwardly generalizes to the case of
real~$\alpha >1$. Thus, we conclude with

\begin{theorem}
For any $\alpha >1$, the number of maximal $\alpha$-gapped repeats
in~$w$ is $O(\alpha n)$.
\label{onPk}
\end{theorem}

Note that the bound of Theorem~\ref{onPk} is asymptotically tight. To see this, it is enough to
consider word $w_k=(0110)^k$. It is easy to check that for a big
enough~$\alpha$ and $k=\Omega (\alpha)$, $w_k$ contains $\Theta (\alpha
|w_k|)$ maximal $\alpha$-gapped repeats whose copies are single-letter words.

\bigskip
We now use Theorem~\ref{onPk} to obtain an upper bound on the number of
maximal $\delta$-subrepetitions.
The following proposition, shown in \cite[Proposition~3]{KPPHr13}, follows from the fact that each maximal $\delta$-subrepetition defines at least one maximal $1/\delta$-gapped repeat (cf. Introduction).

\begin{proposition}[\cite{KPPHr13}]
For $0<\delta <1$, the number of maximal $\delta$-subrepetitions in a string
is no more then the number of maximal $1/\delta$-gapped repeats.
\label{relforep}
\end{proposition}

Theorem~\ref{onPk} combined with Proposition~\ref{relforep} immediately
imply the following upper bound for maximal $\delta$-subrepetitions
that improves the bound of \cite{KKOch} by a $\log n$ factor.

\begin{theorem}
For $0<\delta <1$, the number of maximal $\delta$-subrepetitions in~$w$ is $O(n/\delta)$.
\label{ondelrep}
\end{theorem}

The $O(n/\delta)$ bound on the number of maximal
$\delta$-subrepetitions is asymptotically tight, at least on an unbounded alphabet : word
$\texttt{ab}_1\texttt{ab}_2\ldots \texttt{ab}_k$
contains $\Omega (n/\delta)$  maximal
$\delta$-subrepetitions for $\delta\le 1/2$.

\section{Computing all maximal $\alpha$-gapped repeats}
\label{algorithm}

In this section, we present an $O(\alpha n+S)$ algorithm for computing
all maximal $\alpha$-gapped repeats in a word $w$. This bound has been recently
announced in \cite{Tanimuraetal}, here we present a different
solution. Together with the the $O(\alpha n)$ bound of Theorem~\ref{onPk}, this implies an $O(\alpha n)$-time
algorithm.

\subsection{Computing PR-repeats}
Some maximal $\alpha$-gapped repeats can be specifically located as defined below within maximal
repetitions (runs). For example, word $\texttt{cabababababaa}$ contains maximal
gapped repeats $(\texttt{a})\texttt{babababab}(\texttt{a})$,
 $(\texttt{aba})\texttt{babab}(\texttt{aba})$ and $(\texttt{ababa})\texttt{b}(\texttt{ababa})$
 within the run $\texttt{abababababa}=(\texttt{ab})^{11/2}$. In this section, 
 we describe the structure of such
repeats, and in particular those of them which are periodic (see Section~\ref{counting}), like the repeat
$(\texttt{ababa})\texttt{b}(\texttt{ababa})$ above. We show how those maximal
 $\alpha$-gapped repeats can be extracted from the runs.
Repeats which are located within runs but are not periodic will be found separately,
together with repeats (periodic or not) which are not located within runs. This part
will be described in the next section.

Let $\rep\equiv (u', u'')$ be a periodic gapped repeat.  If the
extensions of $u'$ and $u''$ are the same repetition~$r$ then we say that $r$
{\it generates} $\rep$ and we call $\rep$ {\it PR-repeat}
(abbreviating from {\em Periodic Run-generated}). Gapped repeats which are not PR-repeats are called {\it
non-PR repeats}.
We will use the following fact.

\begin{proposition}
Let $\rep\equiv (u', u'')$ be a maximal gapped repeat such that its copies
$u'$ and $u''$ contain a pair of corresponding factors having the same
extension~$r$. Then $\rep$ is generated by~$r$. \label{samext}
\end{proposition}

\begin{proof}
Observe that to prove the proposition, it is enough to show that both
copies $u'$ and $u''$ are contained in~$r$, i.e. $\beg(r)\le \beg(u')$ and $\endd(r)\ge \endd(u'')$. Let $\beg(r)>\beg(u')$. Then both letters $w[\beg(r)-1]$ and
$w[\beg(r)-1+\per(r)]$ are contained in $u'$. Let these letters be
respectively $j$-th and ($j+\per(r)$)-th letters of $u'$. Then we have
$u''[j]=u'[j]\neq u'[j+\per(r)]=u''[j+\per(r)]$, i.e. $u''[j]\neq
u''[j+\per(r)]$, which is a contradiction to the fact that both letters
$u''[j]$ and $u''[j+\per(r)]$ are contained in~$r$. Relation $\endd(r)\ge \endd(u'')$ is proved analogously.
\end{proof}

All maximal PR-repeats
can be easily computed according to
the following lemma.

\begin{lemma}
A maximal gapped periodic repeat $\rep\equiv (u', u'')$ is generated by a maximal repetition~$r$
if and only if $p(\rep)$ is divisible by $\per(r)$ and
$$
\begin{array}{c}
|r|/2<p(\rep)\le |r|-2\,\per(r),\\
u'\equiv w[\beg(r) \dd \endd(r)-p(\rep)],\\
u''\equiv w[\beg(r)+\per(r)  \dd  \endd(r)]. \label{compgenreps}
\end{array}
$$
\end{lemma}

\begin{proof}
Let $\rep$ be generated by~$r$. Consider prefixes of $u'$ and $u''$ of
length $\per(r)$. These prefixes are equal cyclic roots of~$r$, and by
Proposition~\ref{oncycroot} the difference $\beg(u'')-\beg(u')=p(\rep)$ is divisible by~$\per(r)$. Inequalities
$|r|/2<p(\rep)\le |r|-2\per(r)$ follow immediately from the definition
of a repeat generated by a repetition. To prove the last two conditions of
the lemma, it is sufficient to prove $\beg(u')=\beg(r)$ and $\endd(u'')=\endd(r)$. Let $\beg(u')\neq \beg(r)$, i.e. $\beg(u')>\beg(r)$. Then both letters $w[\beg(u')-1]$ and
$w[\beg(u'')-1]$ are contained in~$r$. Thus, since the difference
$(\beg(u'')-1)-(\beg(u')-1)=p(\rep)$ is divisible by
$\per(r)$, we have  $w[\beg(u')-1]=w[\beg(u'')-1]$ which
contradicts the maximality of~$\rep$. The relation $\endd(u'')=\endd(r)$ is proved analogously. Thus, all the conditions of
the lemma are proved. On the other hand, if $\rep$ satisfies all the
conditions of the lemma then $\rep$ is obviously generated by~$r$.
\end{proof}

\begin{corollary}
A maximal repetition~$r$ generates no more than $\exp(r)/2$ maximal PR-repeats, and
all these repeats can be computed from~$r$ in $O(\exp(r))$ time.
\label{corrgenreps}
\end{corollary}

To find all maximal $\alpha$-gapped PR-repeats in a string $w$,
we first compute all maximal repetitions in~$w$ in $O(n)$ time (see Introduction). Then, for each maximal
repetition~$r$, we output all maximal $\alpha$-gapped repeats generated by~$r$. Using
Corollary~\ref{corrgenreps}, this can be done in $O(\exp(r))$ time. Thus the
total time of processing all maximal repetitions is $O({\cal E}(w))$. Since
${\rm E}(w)=O(n)$ by Theorem~\ref{sumexp}, all maximal
$\alpha$-gapped PR-repeats in~$w$ can be computed in $O(n)$ time.

\subsection{Computing non-PR repeats}

We now turn to the computation of maximal non-PR $\alpha$-gapped
repeats. Recall that non-PR repeats are those which
are either non-periodic, or periodic but not located within a single
run. Our goal is to show that all maximal non-PR $\alpha$-gapped
repeats can be found in $O(\alpha n)$ time.
%
%
Observe that there exists a trivial algorithm for computing all
maximal $\alpha$-gapped repeats in
$O(n^2)$ time that proceeds as follows: for each period~$p\le n$, find
all maximal $\alpha$-gapped repeats with period $p$ in $O(n)$ time by
consecutively comparing symbols $w[i]$ and $w[i+p]$ for $i=1, 2,\ldots , n-p$.

From the results of \cite{Brodal00}, it follows that
all maximal $\alpha$-gapped repeats can be found in time $O(n\log n
+S)$.
This, together with Theorem~\ref{onPk}, implies an $O(\alpha n)$-time
algorithm for the case  $\alpha \ge \log n$.
Therefore, we only have to consider the case $\alpha <\log n$.


\subsubsection*{(i) Preliminaries}
Assume that $\alpha <\log n$. For this case, we proceed with a
modification of the algorithm of~\cite{GabrMan}.
We compute all maximal $\alpha$-gapped non-PR repeats $\rep$
in~$w$ such that $c(\rep)\ge\log n$. To do this, we divide
$w$ into {\it blocks} of $\Delta=(\log n)/4$ consecutive symbols of~$w$.
Without loss of generality, we assume that
$n=2^k\Delta$, i.e. $w$ contains exactly $2^k$ blocks. A word~$x$ of
length~$2^l\Delta$ where $0\le l\le k-1$ is called a {\it basic factor}
of~$w$ if $x= w[i\Delta +1  \dd  (i+2^l)\Delta]$ for some~$i$. Such
an occurrence $w[i\Delta +1  \dd  (i+2^l)\Delta]$ of~$x$ starting at a
block frontier
will be called {\it aligned}. A
basic factor $x$ of length ~$2^l\Delta$, where $1\le l\le k-1$, is called
{\it superbasic} if $x= w[i2^l\Delta +1  \dd  (i+1)2^l\Delta]$ for
some~$i$.
Note that $w$
contains $O(n)$ aligned occurrences of basic factors and $O(\frac{n}{\log
n})$ aligned occurrences of superbasic factors. Let $z\equiv w[q2^l\Delta
+1  \dd  (q+1)2^l\Delta]$ be an aligned occurrence of superbasic factor of
length $2^l$ in~$w$. For $\tau =0,1,\ldots \Delta -1$, an occurrence
$w[q2^l\Delta +1 +\tau  \dd  (q2^l+2^{l-1})\Delta +\tau ]$ of a basic factor
of length $2^{l-1}\Delta$ is called {\it $\tau$-associated} (or simply
{\em associated}) with $z$.
Note that any basic factor occurrence $\tau$-associated
with $z$ is entirely contained in~$z$ and is
uniquely defined by~$z$ and~$\tau$. Thus, $z$ has no more than $\Delta$
associated occurrences of basic factors.

To continue, we need one more definition : for $1\le i, j\le n$, denote by $LCP(i,
j)$ the length of the longest common prefix of $w[i  \dd  n]$ and $w[j  \dd
n]$, and by $LCS(i, j)$ the length of the longest common suffix of $w[1  \dd
i]$ and $w[1  \dd  j]$.

Let $\rep\equiv (u', u'')$ be a maximal gapped repeat in~$w$ such that
$c(\rep )\ge\log n= 4\Delta$. Note that in this case, the left copy $u'$
contains at least one aligned occurrence of superbasic factors. Consider
aligned occurrences of superbasic factors of maximal length contained
in~$u'$. Note that $u'$ can contain either one or two adjacent such
occurrences. Let $z$ be the leftmost of them. Note that in
this case, we have the following restrictions imposed on $u'$:
\begin{equation}
\begin{array}{c}
\beg (z)-|z|<\beg (u')\le\beg (z),\\
\endd (z)\le\endd (u')<\endd (z)+2|z|.
\end{array}
\label{leftcond}
\end{equation}
Thus, $c(\rep )<4|z|$. Consider factor~$z''$ in~$u''$ corresponding
to~$z$ in~$u'$. Note that $z''$ can be non-aligned.
Consider in~$z''$ the leftmost aligned basic factor $y''$ of
of length~$|z''|/2$. Observe that $\beg (z'')\le \beg
(y'')<\beg (z'')+\Delta$ and $y''$ is entirely contained in $z''$.
Let $y'$ be the factor of~$z$ corresponding to factor $y''$ in $z''$.
It is easily seen that $y'$ is an occurrence of a basic factor associated with
$z$, and $\rep$ is uniquely defined by $z$, $y'$ and $y''$. Thus, any
maximal gapped repeat $\rep$ such that $c(\rep )\ge \log n$ is
uniquely defined by a triple $(z, y', y'')$, where $z$ is an aligned
occurrence of some superbasic factor, $y'$ is an occurrence of some basic
factor associated with~$z$, and $y''$ is an aligned occurrence of the same
basic factor. From now on, we will say in such case that $\rep$ {\em is defined} by the
triple $(z, y', y'')$.

Observe that $\rep\equiv (u', u'')$ can be retrieved from $(z, y',
y'')$ using $LCP$ and $LCS$ functions.
\begin{equation}
\begin{array}{c}
\beg (u')=\beg (y')-LCS(\beg (y')-1, \beg (y'')-1),\\
\endd (u')=\endd (y')+LCP(\endd (y')+1, \endd (y'')+1),\\
\beg (u'')=\beg (y'')-LCS(\beg (y')-1, \beg (y'')-1),\\
\endd (u'')=\endd (y'')+LCP(\endd (y')+1, \endd (y'')+1).
\end{array}
\label{compgenrep}
\end{equation}
Assume additionally that $\rep$ is an $\alpha$-gapped repeat for $\alpha
>1$. Then, taking into account inequalities~(\ref{leftcond}) and $c(\rep
)<4|z|$, we have
\begin{eqnarray*}
\endd (y'')&\le &\endd (u'')=\endd (u')+\per(\rep )<\endd (z)+2|z|+\alpha c(\rep )\\
&<&\endd (z)+2|z|+4\alpha |z|<\endd (z) +6\alpha |z|=\endd (z) +12\alpha |y''|.
\end{eqnarray*}
On the other hand, $\beg (y'')\ge \beg (u'')>\endd (u')\ge{\rm
end} (z)$. Thus, for any triple $(z, y', y'')$ defining a maximal
$\alpha$-gapped repeat in~$w$ the occurrence $y''$ is contained in the
segment $w[\endd (z)+1  \dd  \endd (z) +12\alpha |y''|]$ of length
$12\alpha |y''|$ to the right of~$z$. We will denote this segment by ${\cal
I}(z)$. The main idea of the algorithm is to consider all triples $(z, y',
y'')$ which can define maximal $\alpha$-gapped non-PR repeats and
for each such triple, check if it actually defines one, which is then
computed and output.
All the triples $(z, y', y'')$ are considered in a
natural way: for each aligned occurrence~$z$ of a superbasic factor and each
occurrence $y'$ of a basic factor associated with~$z$, we consider all aligned
occurrences $y''$ of the same basic factor in the segment ${\cal I}(z)$.

\subsubsection*{(ii) Naming basic factors on a suffix tree and
  computing their associated occurrences}
We now describe how this computation is implemented.
First we construct a suffix tree for the input string $w$. Suffix tree is a classical
data structure of size $O(n)$ which can be constructed in $O(n)$ for a
word over constant alphabet
see e.g. \cite{Gusfield97}. Using
the suffix tree, we can make in $O(n)$-time preprocessing
which allows to retrieve $LCP(i, j)$ for any $i, j$ in constant time, see
e.g. \cite{Gusfield97}. Similarly, we precompute~$w$ to support
$LCS(i, j)$ for any $i, j$ in constant time. Then we compute all
basic factors of~$w$. This computation is performed by naming all the basic
factors, i.e. assigning to each aligned occurrence of a basic factor a
name of this factor. The most convenient way to name basic factors is
to assign to a basic factor~$y$ of length~$2^l$ a pair $(l, i)$, where $i$ is
the start position of the leftmost aligned occurrence of~$y$ in~$w$. Note
that since we have only $n/\Delta$ distinct start positions~$i$, the size of
the two-dimensional array required for working with these pairs is $O(n)$. To
perform the required computation, we first mark in the suffix tree each node
labeled by a basic factor by the name of this factor (in the case when
this node is implicit we make it explicit). To this end, for each
node~$v$ of the suffix tree we compute the value ${ \minleaf} (v)$ which is
the smallest leaf number divisible by $\Delta$ in the subtree rooted in~$v$
if such a number exists. This can be easily done in $O(n)$ time by a bottom-up
traversal of the tree. Then, each
suffix tree edge $(u, v)$ such that the string depth of~$u$ is less than
$2^l$, the string depth of~$v$ is not less than $2^l$, and  ${ \minleaf} (v)$
is defined is treated in the following way: if the string depth of~$v$ is
$2^l$, node $v$ is marked by name $(l, { \minleaf} (v))$, otherwise a
new node of string depth $2^l$ is created within edge $(u, v)$ and
marked by name $(l, { \minleaf} (v))$. The obtained tree will be
called {\it marked suffix tree}. Since we have $O(n)$ distinct basic factors,
the marked suffix tree contains no more than $O(n)$ additionally inserted nodes. Thus,
this tree has $O(n)$ size and is constructed in $O(n)$ time.

To assign to each aligned occurrence $w[i  \dd  i+2^l-1]$ of a basic factor the name of this factor, we perform
a depth-first top-down traversal of the marked suffix tree. During the
traversal we maintain an auxiliary array ${ basancestor}$: at the first
visit of a node marked by a name $(l, m)$ we set ${ basancestor} [l]$
to $m$,
and at the second visit of this node we reset ${ basancestor} [l]$ to undefined. While
during the traversal we get to a leaf $i$ divisible by $\Delta$,
for each $l=0, 1,\ldots , k-1$ we identify $w[i  \dd  i+2^l-1]$ as an
occurrence of the basic factor named by $(l, { basancestor}[l])$. Note that
this traversal is performed in $O(n)$ time.

Then, we compute all
occurrences of basic factors associated with aligned occurrences of
superbasic factors. This is done again by a depth-first
top-down traversal of the marked suffix tree. During the traversal, we
maintain the same auxiliary array ${ basancestor}$. Assume that during the traversal
we get to a leaf labelled by a position
$q2^p\Delta
+1+\tau$, where $q$ is odd and $0\le\tau <\Delta$. Then for each $l=0, 1,\ldots
, p-1$ such that ${ basancestor} [l]$ is defined, we identify $w[q2^p\Delta
+1+\tau  \dd  (q2^p+2^l)\Delta +\tau]$ as an occurence of the basic factor
named $(l, { basancestor} [l])$, which is $\tau$-associated with the
superbasic
factor occurrence $w[q2^p\Delta +1  \dd  (q2^p+2^{l+1})\Delta ]$. Observe
that this traversal is performed in $O(n)$ time as well.

\subsubsection*{(iii) Computing lists of aligned occurrences of basic factors}
Let $y$ be a $\Delta$-periodic basic factor (cf Introduction). Note that $y$ is also periodic,
and then any occurrence of~$y$ in~$w$ is a repetition. By
Proposition~\ref{divminper}, the period~$\per(y)$ is a divisor of
$p_{\Delta}(y)$. Given the value $p_{\Delta}(y)$, we can compute in constant time
the extension~$r$ of any occurrence~$y'$ of a $\Delta$-periodic basic
factor~$y$ as follows:
\begin{eqnarray*}
\beg (r)&=&\beg (y')-LCS(\beg (y')-1, \beg (y')+p_{\Delta}(y)-1),\\
\endd (r)&=&\endd (y')+LCP(\beg (y')+1, \beg (y')-p_{\Delta}(y)+1).
\end{eqnarray*}
Using Proposition~\ref{minDeltaper}, it is easy to show that any set of all
aligned occurrences of~$y$ having the same extension is a sequence of
occurrences, where the difference between start positions of any two
consecutive occurrences is equal to $p_{\Delta}(y)$, i.e. the start positions
of all these occurrences form a finite arithmetic progression with common
difference $p_{\Delta}(y)$. We will call these sets {\it runs of
occurrences}. The following fact can be easily proved.
\begin{proposition}
Let $y'$, $y''$ be two consecutive aligned occurrences of a basic factor~$y$
in~$w$. Then $|\beg (y')-\beg (y'')|\le |y|/2$ if and only if $y$
is $\Delta$-periodic, $y'$ and $y''$ are contained in the same run of
occurrences, and, moreover, $|\beg (y')-\beg (y'')|=p_{\Delta}(y)$.
\label{occurrun}
\end{proposition}

At the next step of the algorithm, in order to effectively select appropriate
occurrences~$y''$ in the checked triples $(z, y', y'')$, for each basic
factor~$y$ we construct a linked list ${\alignocc} (y)$ of all aligned
occurences of~$y$ in the left-to-right order in~$w$. If $y$ is not
$\Delta$-periodic, each item of ${\alignocc} (y)$ consists of only one aligned
occurrence of~$y$ defined, for example, by its start position (we will call
such items {\it ordinary}). If $y$ is $\Delta$-periodic, each item of
${\alignocc} (y)$ contains a run of aligned occurrences of~$y$. If a run of
aligned occurrences of~$y$ consists of only one occurrence, we will consider
the item of ${\alignocc} (y)$ for this run as ordinary, otherwise, if a run of
aligned occurrences of~$y$ consists of at least two occurrences, the item of
${\alignocc} (y)$ for this run will be defined, for example, by start
positions of leftmost and rightmost occurrences in the run and the value
$p_{\Delta}(y)$ (such item will be called {\it runitem}). The following fact
follows from Proposition~\ref{occurrun}.

\begin{proposition}
Let $y'$, $y''$ be two consecutive aligned occurrences of a basic factor~$y$
in~$w$. Then $|\beg (y')-\beg (y'')|\le |y|/2$ if and only if  $y'$
and $y''$ are contained in the same runitem of ${\alignocc} (y)$ and,
moreover, $|\beg (y')-\beg (y'')|=p_{\Delta}(y)$. \label{occuritem}
\end{proposition}

Proposition~\ref{occuritem} implies that if two aligned occurrences $y'$,
$y''$  of a basic factor~$y$ are contained in distinct items of ${\alignocc}
(y)$ then $|\beg (y')-\beg (y'')|>|y|/2$. Therefore, we have the following
consequence from the proposition.

\begin{corollary}
Let $y$ be a basic factor of~$w$. Then for any segment~$v$ in~$w$, the list
${\alignocc} (y)$ contains $O(|v|/|y|)$ items having at least one occurrence
of~$y$ contained in~$v$. \label{itemsnumber}
\end{corollary}

To construct the lists ${\alignocc}$, for each $i=1, 2, \ldots, n$ and each
$l=0, 1,\ldots , k-1$, we insert consecutively the occurrence $y'\equiv w[i
\dd  i+2^l-1]$ of some basic factor~$y$ to the appropriate list ${\alignocc}
(y)$ as follows. Consider the last item in the current list ${\alignocc} (y)$.
Let it be an ordinary item consisting of an occurrence $y''$ of~$y$ starting
at position~$j$. Denote $\delta =i-j$. Consider the following two cases
for~$\delta$. Let $\delta >|y|/2$. Then, by Proposition~\ref{occuritem},
$y''$ and $y'$ are contained in distinct items of ${\alignocc} (y)$, and in
this case we insert $y'$ to ${\alignocc} (y)$ as a new ordinary item. Now let
$\delta\le|y|/2$. In this case, by Proposition~\ref{occuritem}, $y''$ and
$y'$ are the first two occurrences of the same run of occurrences of~$y$ and,
moreover, $\delta =p_{\Delta}(y)$. Let $r$ be the extension of the
occurrences of this run. It is easy to see that
$$
\endd (r) = \endd (y') + LCP(\endd (y'') +1, \endd (y') +1),
$$
i.e. $\endd (r)$ can be computed in constant time. From the values ${\rm
beg} (y'')$, $\endd (r)$ and $p_{\Delta}(y)$ we can compute
in constant time the start position of the last occurrence of~$y$ in the
considered run of occurrences and thereby identify completely this run.
Thus, in this case we replace the last item of ${\alignocc} (y)$ by the
identified run of occurrences of~$y$. Now let the last item in ${\alignocc}
(y)$ be a run of occurrences. Then, if $y'$ is not contained in this run, we
insert $y'$ to ${\alignocc} (y)$ as a new ordinary item. Thus, each occurrence
of a basic factor in~$w$ is processed in constant time, and the total time for
construction of lists ${\alignocc}$ is $O(n)$.

Furthermore, in order to optimize the selection of appropriate occurrences $y''$ in
the checked triples $(z, y', y'')$, for each pair $(z, y')$ where $z$ is an
aligned occurrence of a superbasic factor and $y'$ is an occurrence of some
basic factor $y$ associated with~$z$, we compute a pointer ${\firstocc} (z,
y')$ to the first item in ${\alignocc} (y)$ containing at least one occurrence
of~$y$ to the right of~$z$. For these purposes, we use auxiliary lists
${\factends} (i)$ defined for each position~$i$ in~$w$. Lists ${\factends} (i)$
consist of pairs $(z, y')$ and are constructed at the stage computation of
occurrences associated with aligned occurrences of superbasic factors: each
time we find a new occurrence~$y'$ associated with an aligned occurrence~$z$
of a superbasic factor, we insert the pair $(z, y')$ into the list ${\factends} ({\rm
end} (z)+1)$. After construction of lists ${\alignocc}$, we compute
consecutively for each $i=1, 2,\ldots , n$   pointers ${\firstocc} (z, y')$
for all pairs $(z, y')$ from the list ${\factends} (i)$. During the
computation, we save in each list ${\alignocc} (y)$ the last item pointed before
(this item is denoted by ${lastpnt} (y)$). To compute ${\firstocc} (z, y')$,
we go through the list ${\alignocc} (y)$ from ${lastpnt} (y)$ (or from the
beginning of ${\alignocc} (y)$ if ${lastpnt} (y)$ does not exist) until we
find the first item containing at least one occurrence of~$y$ to the right of
the position~$i$. The found item is pointed by ${\firstocc} (z, y')$ and
becomes a new item ${lastpnt} (y)$. Since the total size of lists
${\alignocc}$ and ${\factends}$ is $O(n)$, the total time of computing
${\firstocc} (z, y')$ is also $O(n)$.

\subsubsection*{(iv) Main step: computing large repeats}

At the main stage of the algorithm, in order to process each pair $(z, y')$,
note that all appropriate for $(z, y')$ occurrences~$y''$ contained in ${\cal
I}(z)$ are located in the fragment of ${\alignocc} (y)$ consisting of all
items having at least one occurrence of~$y$ contained in~${\cal I}(z)$. We
will call this fragment {\it checked fragment}. Thus, we consider all items of
the checked fragment by going through this fragment from the first item which
can be found in constant time by the value ${\firstocc} (z, y')$. For each
considered item, we check triples $(z, y', y'')$ for all occurrences $y''$
from this item as follows.

Let the considered item be an ordinary item consisting of only one
occurrence~$y''$. Recall that gapped repeat $(u', u'')$ defined by the
triple $(z, y', y'')$ can be computed in constant time by
formulas~(\ref{compgenrep}). Thus, if $(u', u'')$ is an $\alpha$-gapped
repeat satisfying conditions~(\ref{leftcond}), we output it.

Now let the item considered in the checked fragment be a runitem. This implies
that basic factor~$y$ is $\Delta$-periodic, i.e $y$ is $\Delta$-periodic.
Moreover, from the runitem we can derive the value $p_{\Delta}(y)$. Therefore we can
compute in constant time extensions $r'$ and~$r''$ of occurrences
$y'$ and~$y''$ respectively. Denote by~$\rho$ the run of occurrences
contained in the runitem. Recall that our goal is to compute effectively all
$\alpha$-gapped repeats defined by triples $(z, y', y'')$ such that
$y''\in\rho$. Note that, if  $r'$ and~$r''$ are the same repetition, then by
Proposition~\ref{samext} all such repeats are PR-repeats, therefore we can assume that
$r'$ and~$r''$ are distinct repetitions. Let $(u', u'')$ be an
$\alpha$-gapped repeat defined by a triple $(z, y', y'')$ where $y''\in\rho$.
First, consider the case when $u'$ is not contained in $r'$, i.e. either ${\rm
beg} (u')<\beg (r')$ or $\endd (u')>\endd (r')$.

\begin{proposition}
If $\beg (u')<\beg (r')$, then
$\beg (r')-\beg (u')=\beg (r'')-\beg (u'')$.
\label{begdif}
\end{proposition}
\begin{proof}
Define $\gamma'=\beg (r')-\beg (u')$, $\gamma''=\beg
(r'')-\beg (u'')$. Let $\gamma' >\gamma''$. Then $u'[\gamma'
+\per(y)]\neq u'[\gamma']=u''[\gamma']=u''[\gamma' +\per(y)]$, i.e. we have
a contradiction $u'[\gamma' +\per(y)]\neq u''[\gamma'
+\per(y)]$. Similarly, we obtain
a contradiction $u'[\gamma'' +\per(y)]\neq u''[\gamma''
+\per(y)]$ in the case $\gamma' <\gamma''$.
\end{proof}

The following proposition can be proved analogously.

\begin{proposition}
If $\endd (u')>\endd (r')$, then
$\endd (u')-\endd (r')=\endd (u'')-\endd (r'')$.
\label{enddif}
\end{proposition}

Define
\begin{eqnarray*}
s_\mathit{left}&=&\beg (y')+(\beg (r'')-\beg (r')),\\
s_\mathit{right}&=&\beg (y')+(\endd (r'')-\endd (r')).
\end{eqnarray*}
From Propositions~\ref{begdif} and~\ref{enddif}, we derive the following fact.

\begin{corollary}
If $\beg (u')<\beg (r')$ then $\beg (y'')=s_\mathit{left}$.
If $\endd (u')>\endd (r')$ then $\beg (y'')=s_\mathit{right}$.
\end{corollary}

Thus, for computing $\alpha$-gapped repeats $(u', u'')$ such that $u'$ is not
contained in $r'$, it is enough to consider in $\rho$ only occurrences
$y''_\mathit{left}$ and $y''_\mathit{right}$ with start positions $s_\mathit{left}$ and
$s_\mathit{right}$ respectively, provided that these occurrences exist. We check the occurrences
$y''_\mathit{left}$ and $y''_\mathit{right}$ in the same way as we did
for occurrence $y''$ in the
case of ordinary item. Then, it remains to check all occurrences from $\rho$
except for possible occurrences $y''_\mathit{left}$ and $y''_\mathit{right}$. Denote by
$\rho'=\rho\setminus\{y''_\mathit{left}, y''_\mathit{right}\}$ the set of all such
occurrences. Assume that $|r'|\le |r''|$, i.e. $s_\mathit{left}\le s_\mathit{right}$ (the
case $|r'|>|r''|$ is similar). In order to check
all occurrences from $\rho'$, we consider the following subsets of
$\rho'$ separately: subset $\rho'_1$ of all occurrences $y''$ such that $\beg
(y'')<s_\mathit{left}$, subset $\rho'_2$ of all occurrences $y''$ such that
$s_\mathit{left}<\beg (y'')<s_\mathit{right}$, and subset $\rho'_3$ of all
occurrences $y''$ such that $s_\mathit{right}<\beg (y'')$. Note that start
positions of all occurrences in each of these subsets form a finite
arithmetic progression with common difference $p_{\Delta}(y)$. Thus,
we unambiguously denote all occurrences in each of the subsets
$\rho'_i$,  $i=1, 2, 3$, by $y''_0, y''_1,\ldots , y''_k$ where $y''_0$ is
the leftmost occurrence in the subset $\rho'_i$ and $\beg (y''_j)= {\beg} (y''_0)+jp_{\Delta}(y)$ for $j=1,\ldots , k$. Note that values ${\beg} (y''_0)$ and~$k$ for each subset  $\rho'_i$ can be computed in
constant time.

First, consider an occurrence~$y''_j$ from $\rho'_1$. Let $\rep\equiv (u',
u'')$ be the repeat defined by triple $(z, y', y''_j)$. Note that
\begin{equation}
\per(\rep )=\beg (y''_j)-\beg (y')=q+jp_{\Delta}(y),
\label{eqvforp}
\end{equation}
where $q=\beg (y''_0)-\beg (y')$. Taking into account that $y'$ and
$y''_j$ are contained in maximal repetitions $r'$ and $r''$ respectively,
it is easy to verify that
$$
\begin{array}{c}
LCS(\beg (y')-1, \beg (y''_j)-1)=\beg (y''_j)-\beg (r''),\\
LCP(\endd (y')+1, \endd (y''_j)+1)=\endd (r')-\endd (y').
\end{array}
$$
Therefore, $\beg (u')=\beg (r'')-\per(\rep )=q'-jp_{\Delta}(y)$,
where $q'=\beg (r'')-q$, and $\endd (u')=\endd (r')$. It follows that
$$
c(\rep )=|u'|=\endd (u')-\beg (u')+1=q''+jp_{\Delta}(y),
$$
where $q''=\endd (r')+1-q'$. Recall that for any $\alpha$-gapped repeat
$\rep$, we have $c(\rep)<\per(\rep)\le\alpha
c(\rep)$. Thus, $\rep$ is an $\alpha$-gapped repeat if and only if
\begin{equation}
q''<q\le\alpha q''+(\alpha -1)jp_{\Delta}(y).
\label{alphacond1}
\end{equation}
Moreover, $u'$ has to satisfy conditions~(\ref{leftcond}). Thus, the triple
$(z, y', y''_j)$ defines an $\alpha$-gapped repeat if and only if conditions
(\ref{alphacond1}) and~({\ref{leftcond}) are verified for~$j$. Note that all
these conditions are linear inequalities on~$j$, and then can be resolved
in constant time. Thus, we output all $\alpha$-gapped repeats defined by
triples $(z, y', y'')$ such that $y''\in\rho'_1$ in time $O(1+S)$, where $S$
is the size of the output.

Now consider an occurrence~$y''_j$ from $\rho'_2$. Let $\rep\equiv (u',
u'')$ be the repeat defined by the triple $(z, y', y''_j)$. Note that in this
case, $\per(\rep)$ also satisfies relation~(\ref{eqvforp}). Analogously
to the previous case of set $\rho'_1$, we obtain that $\beg
(u')=\beg (r')$ and $\endd (u')=\endd (r')$, and then
$c(\rep)=|r'|$. Therefore, $\rep$ is an $\alpha$-gapped repeat if and
only if
\begin{equation}
|r'|<q+jp_{\Delta}(y)\le \alpha |r'|.
\label{alphacond2}
\end{equation}
Thus, in this case, we output all $\alpha$-gapped repeats defined by triples
$(z, y', y''_j)$ such that $j$ satisfies conditions (\ref{alphacond2})
and~({\ref{leftcond}). Since all these conditions can be resolved for~$j$ in
constant time, all these repeats can be output in time $O(1+S)$ where $S$
is the size of output.

Finally, consider an occurrence~$y''_j$ from $\rho'_3$. Let $\rep\equiv (u',
u'')$ be the repeat defined by triple $(z, y', y''_j)$. In this case,
$\per(\rep)$ also satisfies relation~(\ref{eqvforp}). Analogously to the
case of set $\rho'_1$, we obtain that $\beg (u')=\beg (r')$
and $\endd (u')=\endd (r')-\per(\rep )= \widehat q'-jp_{\Delta}(y)$,
where $\widehat q'=\endd (r'')-q$, and then
$$
c(\rep)=\endd (u')-\beg (u')+1=\widehat q''-jp_{\Delta}(y),
$$
where $\widehat q''=\widehat q'-\beg (r')+1$. Therefore, $\rep$ is an
$\alpha$-gapped repeat if and only if
\begin{equation}
\widehat q''-jp_{\Delta}(y)<q+jp_{\Delta}(y)\le \alpha (\widehat q''-jp_{\Delta}(y)).
\label{alphacond3}
\end{equation}
Thus, in this case, we output all $\alpha$-gapped repeats defined by triples
$(z, y', y''_j)$ such that $j$ satisfies conditions (\ref{alphacond3})
and~({\ref{leftcond}). Like in the previous cases, this can be done in time
$O(1+S)$, where $S$ is the size of the output.

Putting together all the considered cases, we conclude that all $\alpha$-gapped
repeats defined by triples $(z, y', y'')$ such that $y''\in\rho$ can be
computed in time $O(1+S)$ where $S$ is the size of output. Thus, in $O(1+S)$
time we can process each item of the checked fragment. Therefore, since by
Corollary~\ref{itemsnumber} the checked fragment has $O(\alpha)$ items, the
total time for processing pair $(z, y')$ is $O(\alpha +S)$ where $S$ is
the total number of $\alpha$-gapped repeats defined by triples $(z, y',
y'')$. Since each occurrence~$z$ has no more than $\Delta$ associated
occurrences $y'$, the total number of processed pairs $(z, y')$ is $O(n)$.
Thus the time complexity of the main stage of the algorithm is $O(\alpha n
+S)$, where $S$ is the size of the output. Taking into account that $S=O(\alpha
n)$ by Theorem~\ref{onPk}, we conclude that the time complexity of the main
stage is $O(\alpha n)$. Thus, all maximal $\alpha$-gapped non-PR repeats
$\rep$ in~$w$ such that $c(\rep)\ge\log n$ can be computed in $O(\alpha
n)$ time.

\subsubsection*{(v) Computing small repeats}
To compute all remaining maximal $\alpha$-gapped non-PR repeats
in~$w$, note that the length of any such repeat~$\rep$ is not greater than
$$
(1+\alpha)c(\rep)<(1+\log n)\log n<2\log^2 n.
$$
Thus, setting $\Delta'=\lfloor 2\log^2 n\rfloor$, any such repeat is
contained in at least one of segments ${\cal I}'_i\equiv w[i\Delta'+1  \dd
(i+2)\Delta']$ for $0\le i< n/\Delta'$. Therefore, all the remaining
$\alpha$-gapped repeats can be found by searching separately in
segments ${\cal I}'_i$. The procedure of searching for repeats in ${\cal
I}'_i$ is similar to the algorithm described above. If $\alpha\ge
\log\log n$, searching for repeats in ${\cal I}'_i$ can be done by the algorithm
proposed in~\cite{Brodal00}. The $O(|{\cal I}'_i|\log |{\cal I}'_i| +S)$
time complexity implied by this algorithm, where by Theorem~\ref{onPk}
the output size $S$
is $O(\alpha |{\cal I}'_i|)$, can be bounded here by
$O(\alpha \Delta')$. Thus, the total time complexity of the search in
all segments ${\cal I}'_i$ is $O(\alpha n)$. In the case of $\alpha
<\log\log n$, we search in each segment ${\cal I}'_i$ for all remaining
maximal $\alpha$-gapped non-PR repeats $\rep$ in~$w$ such that
$c(\rep )\ge \log |{\cal I}'_i|$ in time $O(\alpha \Delta')$, in the same
way as we described above for the word~$w$. The total time of the search
in all segments ${\cal I}'_i$ is $O(\alpha n)$. Then, it remains to compute
all  maximal $\alpha$-gapped non-PR repeats~$\rep$ in~$w$ such that
$c(\rep)<\log |{\cal I}'_i|\le 3\log\log n$. Note that the length of any
such repeat is not greater than
$$
(1+\alpha)3\log\log n<(1+\log\log n)3\log\log n\le 6\log^2\log n.
$$
Thus, setting $\Delta''=\lfloor 6\log^2\log n\rfloor$, any such repeat
is contained in at least one of the segments ${\cal I}''_i\equiv w[i\Delta''+1
\dd  (i+2)\Delta'']$ for $0\le i< n/\Delta''$. Note that these segments
are words of length $2\Delta''$ over an alphabet of size~$\sigma$, therefore
the total number of distinct segments ${\cal I}''_i$ is not greater than
$\sigma^{2\Delta''}\le\sigma^{12\log^2\log n}$. In each of the distinct
segments ${\cal I}''_i$, all maximal $\alpha$-gapped repeats can be found
by the trivial algorithm described above in $O({\Delta''}^2)=O(\log^4\log
n)$ time. Thus, maximal $\alpha$-gapped repeats in all distinct
segments ${\cal I}''_i$ can be found in $O(\sigma^{12\log^2\log
n}\log^4\log n)=o(n)$ time. We conclude that all remaining maximal $\alpha$-gapped
repeats in~$w$ can be found in $O(n+S)$ time where $S$ is the total number
of maximal $\alpha$-gapped repeats contained in all segments ${\cal
I}''_i$. According to Theorem~\ref{onPk}, this number can be bounded by
$O(\alpha n)$, and the time for finding all the remaining maximal
$\alpha$-gapped repeats can be bounded by $O(\alpha n)$ as well. This
leads to the final result.

\begin{theorem}
For a fixed~$\alpha>1$, all maximal $\alpha$-gapped repeats in a word of
length~$n$ over a constant alphabet can be found in $O(\alpha n)$ time.
\label{algteorem}
\end{theorem}

Note that since, as mentioned earlier, a word can contain $\Theta (\alpha
n)$ maximal $\alpha$-gapped repeats, the $O(\alpha
n)$ time bound stated in Theorem~\ref{algteorem} is asymptotically optimal.

\section{Conclusions}

Besides gapped repeats we can also consider gapped palindromes which are
factors of the form $uvu^R$ where $u$ and $v$ are nonempty words and $u^R$ is
the reversal of~$u$ \cite{KK09}. A gapped palindrome $uvu^R$ in a word~$w$ is called {\it
maximal} if $w[\endd (u)+1]\neq w[\beg (u^R)-1]$ and $w[\beg
(u)-1]\neq w[\endd (u^R)+1]$ for $\beg (u)>1$ and $\endd
(u^R)<|w|$. A maximal gapped palindrome $uvu^R$ is $\alpha$-gapped if
$|u|+|v|\le\alpha |u|$ \cite{GabrMan}. It can be shown analogously to the results of this
paper that for $\alpha >1$ the number of maximal $\alpha$-gapped palindromes
in a word of length~$n$ is bounded by $O(\alpha n)$ and for the case of
constant alphabet, all these palindromes can be found in $O(\alpha n)$
time\footnote{Note that in \cite{GabrMan}, the number of maximal
  $\alpha$-gapped palindromes was conjectured to be $O(\alpha^2 n)$.}.

In this paper we consider maximal $\alpha$-gapped repeats with $\alpha >1$.
However this notion can be formally generalized to the case of $\alpha\le 1$.
In particular, maximal $1$-gapped repeats are maximal
repeats whose copies are adjacent or overlapping. It is easy to see that such
repeats form runs whose minimal periods are divisors of the periods of these
repeats. Moreover, each run in a word is formed by at least one maximal
$1$-gapped repeat, therefore the number of runs in a word is not greater than the
number of maximal $1$-gapped repeats. More precisely, each run $r$
is formed by $\lfloor \exp(r)/2\rfloor$ distinct maximal $1$-gapped repeats.
Thus, if a word contains runs with exponent greater than or equal to~4 then
the number of maximal $1$-gapped repeats is strictly greater
than the number of runs. However, using an easy modification of the proof
of ``runs conjecture'' from~\cite{RunsTheor}, it can be also proved the
number of maximal $1$-gapped repeats in a word is strictly less than the
length of the word. Moreover, denoting by ${\cal R}(n)$ (respectively,
${\cal R}_1(n)$) the
maximal possible number of runs (respectively, maximal possible number of maximal
$1$-gapped repeats) in words of length~$n$, we conjecture that ${\cal
R}(n)={\cal R}_1(n)$ since known words with a relatively large number of runs
have no runs with big exponents. We can also consider the case of $\alpha <1$
for repeats with overlapping copies, in particular, the case of maximal
$1/k$-gapped repeats where $k$ is integer greater than~1. It is easy to see
that such repeats form runs with exponents greater than or equal to $k+1$. It
is known from~\cite[Theorem~11]{RunsTheor} that the number of such runs in a
word of length~$n$ is less than $n/k$, and it seems to be possible to modify
the proof of this fact for proving that the number of maximal $1/k$-gapped
repeats in the word is also less than  $n/k=\alpha n$. These observations
together with results of computer experiments for the case of $\alpha >1$
leads to a conjecture that for any $\alpha >0$, the number maximal
$\alpha$-gapped repeats in a word of length~$n$ is actually less than $\alpha
n$. This generalization of the ``runs conjecture'' constitutes an interesting open
problem. Another interesting open question is whether the
obtained $O(n/\delta)$ bound on the number of maximal $\delta$-subrepetitions
is asymptotically tight for the case of constant alphabet.

\paragraph{Acknowledgments.}
This work was partially supported by Russian Foundation for Fundamental
Research (Grant 15-07-03102).

\bibliographystyle{abbrv}
\bibliography{maxrep}

\end{document}